\pdfoutput=1
\documentclass[11pt]{amsart}
\usepackage[margin=1.15in]{geometry}
\usepackage{amsmath,amssymb,amsthm}
\usepackage{graphicx}
\usepackage[hidelinks]{hyperref}

\newtheorem{theorem}{Theorem}
\newtheorem{lemma}[theorem]{Lemma}
\newtheorem{proposition}[theorem]{Proposition}

\theoremstyle{remark}
\newtheorem{remark}[theorem]{Remark}
\newtheorem{example}[theorem]{Example}

\newcommand{\Hur}{\mathcal{H}}
\newcommand{\Nm}{\mathrm{N}}
\newcommand{\tr}{\operatorname{tr}}
\newcommand{\F}{\mathbb{F}}
\newcommand{\legendre}[2]{\left(\tfrac{#1}{#2}\right)}

\title[Conjugation invariants and metacommutation]{Conjugation invariants determine the metacommutation permutation only up to relabelling}
\author{Matthew Fried}
\date{\today}

\begin{document}

\begin{abstract}
Let $\Hur$ be the Hurwitz quaternions, $p$ an odd prime, and $Q\in\Hur$ a prime of norm $q\neq p$. Metacommutation $PQ=Q'P'$ induces a permutation $\pi_Q$ of the $p+1$ left-associate classes of primes of norm $p$. Cohn--Kumar compute its sign and fixed-point count, and Leite--Machiavelo its full cycle structure, by formulas depending on $Q$ only through conjugation-invariant data ($q$ and $\tr Q$). We prove that this is all such data can ever carry: $\pi_Q$ is constant on the unit-conjugacy orbit of $Q$ if and only if $\pi_Q=\mathrm{id}$, because the Hurwitz unit group maps onto a subgroup of $\mathrm{PGL}_2(\F_p)$ with trivial centralizer. Hence, away from the identity, no conjugation-invariant function of $Q$ determines the labelled permutation, or even one specified value $\pi_Q(C)$. The minimal case is made fully explicit: for $(p,q)=(3,5)$ the primes $2+i,\,2+j,\,2+k,\,2-i$ form a single unit-conjugacy orbit yet induce four pairwise distinct $4$-cycles of the same four classes, disagreeing at every class. We further note that splittings $\Hur/p\Hur\cong M_2(\F_p)$ form a torsor under $\mathrm{PGL}_2(\F_p)$, and that by orbit--stabilizer one value $\pi_Q(C)$ is exactly one coset of the stabilizer of $C$. All sixteen refactorizations in the witness are listed in an appendix and verified by machine along two independent routes.
\end{abstract}

\maketitle

\section{Introduction}

Let $\Hur=\mathbb{Z}\langle i,j,\tfrac{1}{2}(1+i+j+k)\rangle$ be the Hurwitz order, with norm $\Nm$, trace $\tr(Q)=Q+\bar Q$, and unit group $\Hur^\times$ of order $24$. Conway and Smith \cite{CS} showed that factorizations of a Hurwitz integer modelling a fixed ordered factorization of its norm are related by unit migration, recombination, and \emph{metacommutation}: primes $P,Q$ of distinct norms $p,q$ satisfy
\[
PQ \;=\; Q'P', \qquad \Nm(Q')=q,\ \Nm(P')=p,
\]
with $(Q',P')$ unique up to unit migration $(Q'u^{-1},uP')$. For $p$ odd, fixing $Q$ and letting $P$ range over primes of norm $p$ gives a permutation $\pi_Q$ of the set $\mathcal{C}_p$ of $p+1$ left-associate classes $[P]=\Hur^\times P$.

Much is known about $\pi_Q$. Cohn and Kumar \cite{CK} proved $\operatorname{sgn}\pi_Q=\legendre{q}{p}$ and that $\pi_Q$ has $1+\legendre{\tr(Q)^2-4q}{p}$ fixed points unless $Q$ is congruent mod $p$ to a rational integer, in which case $\pi_Q=\mathrm{id}$. Forsyth, Gurev, and Shrima \cite{FGS} identified $\pi_Q$ with the action of an element of $\mathrm{PGL}_2(\F_p)$ on $\mathbb{P}^1(\F_p)$; Leite and Machiavelo \cite{LM} determined the cycle structure completely; Babei and Chari \cite{BC} extended the theory to Eichler orders. Every one of these formulas reads off $Q$ only its norm and trace, hence only its unit-conjugacy class: the \emph{cycle type} of $\pi_Q$ is determined by conjugation-invariant data.

This note proves that nothing else is.

\begin{theorem}\label{thm:A}
Let $p$ be an odd prime and $Q$ a Hurwitz prime of norm $q\neq p$. The following are equivalent:
\begin{enumerate}
\item[(i)] $\pi_{uQu^{-1}}=\pi_Q$ for every $u\in\Hur^\times$;
\item[(ii)] $\pi_Q=\mathrm{id}$;
\item[(iii)] $Q\equiv n \pmod{p\Hur}$ for some $n\in\mathbb{Z}$.
\end{enumerate}
Consequently, if $I$ is any function on Hurwitz primes of norm $q$ with $I(uQu^{-1})=I(Q)$ for all units $u$, then $Q\mapsto\pi_Q$ is not a function of $I(Q)$ on any set of primes containing some $Q$ with $\pi_Q\neq\mathrm{id}$ together with its unit conjugates.
\end{theorem}

The proof is short: metacommutation and unit conjugation act on $\mathcal{C}_p$ through one copy of $\mathrm{PGL}_2(\F_p)$ (Lemma~\ref{lem:model}), and the image of the unit group there has trivial centralizer (Lemma~\ref{lem:centralizer}). The failure of invariance is then witnessed in the smallest case in the strongest form:

\begin{theorem}\label{thm:B}
For $(p,q)=(3,5)$, the four Hurwitz primes $2+i,\ 2+j,\ 2+k,\ 2-i$ lie in a single $\Hur^\times$-conjugacy orbit --- so \emph{every} conjugation-invariant function agrees on them --- yet they induce four pairwise distinct $4$-cycles of the four classes of norm-$3$ primes, and their values disagree at every single class.
\end{theorem}

Theorems \ref{thm:A} and \ref{thm:B} complement \cite{CK,LM} rather than compete with them: those results compute the invariant layer exactly; ours show the layer is complete. No computational claim is made --- once a splitting $\Hur/p\Hur\cong M_2(\F_p)$ is fixed, every value of $\pi_Q$ is computed quickly by \cite{FGS}. The content is about \emph{invariants}: what can be known of $\pi_Q$ before the action of $Q$ on a supplied class is evaluated. The equivariance identity underlying the theorem is elementary; the point of this note is the sharp equivalence, the trivial-centralizer mechanism behind it, the explicit minimal witness, and the exact pricing of a single value (Proposition~\ref{prop:cost}).

\section{The matrix model}\label{sec:model}

Fix an odd prime $p$ and a prime $q\neq p$. Since $\Hur$ is a maximal order in a quaternion algebra ramified only at $2$ and $\infty$, reduction gives an $\F_p$-algebra isomorphism $\varphi:\Hur/p\Hur\xrightarrow{\ \sim\ }M_2(\F_p)$; fix one, called a \emph{frame}. For a Hurwitz prime $P$ of norm $p$, the matrix $\varphi(P)$ has rank exactly $1$: its determinant is $\Nm(P)\equiv0$, and rank $0$ would put $P$ in $p\Hur$, forcing $p^2\mid\Nm(P)$. Left unit multiples change $\varphi(P)$ by an invertible left factor, which preserves the row space; so each class $[P]\in\mathcal{C}_p$ has a well-defined line
\[
r_{[P]}\;:=\;\operatorname{row}\varphi(P)\ \in\ \mathbb{P}^1(\F_p),
\]
and $[P]\mapsto r_{[P]}$ is a bijection $\mathcal{C}_p\to\mathbb{P}^1(\F_p)$ \cite{FGS}. For invertible $X\in M_2(\F_p)$, write $R_X$ for the permutation $\ell\mapsto\ell X$ of $\mathbb{P}^1(\F_p)$; then $R_XR_Y=R_{XY}$ (acting left to right), the action factors through $\mathrm{PGL}_2(\F_p)$, and it is faithful, since a M\"obius transformation fixing all $p+1\geq3$ points of $\mathbb{P}^1$ is the identity.

\begin{lemma}\label{lem:model}
In any frame:
\begin{enumerate}
\item[(a)] $\pi_Q$ is right multiplication by $\varphi(Q)$; that is, $r_{\pi_Q[P]}=r_{[P]}\,\varphi(Q)$.
\item[(b)] the conjugation permutation $\rho_u[P]:=[uPu^{-1}]$, $u\in\Hur^\times$, is right multiplication by $\varphi(u)^{-1}$.
\end{enumerate}
In particular $\pi_Q$ is a well-defined bijection of $\mathcal{C}_p$, and $\pi_{\bar Q}=\pi_Q^{-1}$, since $\varphi(\bar Q)=\Nm(Q)\,\varphi(Q)^{-1}$ agrees with $\varphi(Q)^{-1}$ up to scalars.
\end{lemma}

\begin{proof}
(a) $\varphi(Q)$ is invertible since $\det\varphi(Q)\equiv\Nm(Q)=q\not\equiv0$. Write $L=PQ=Q'P'$. Then $\varphi(L)=\varphi(P)\varphi(Q)$ has row space $\operatorname{row}\varphi(P)\cdot\varphi(Q)$, a line. Also $\varphi(L)=\varphi(Q')\varphi(P')$, and every row of a product lies in the row space of the right factor, so $\operatorname{row}\varphi(L)\subseteq\operatorname{row}\varphi(P')$; both are lines, hence they are equal. Thus $r_{[P']}=r_{[P]}\varphi(Q)$; this depended only on $[P]$, and $R_{\varphi(Q)}$ is a bijection, so $\pi_Q$ is one. (b) $\varphi(uPu^{-1})=\varphi(u)\varphi(P)\varphi(u)^{-1}$ has row space $\operatorname{row}\varphi(P)\cdot\varphi(u)^{-1}$.
\end{proof}

\begin{lemma}\label{lem:centralizer}
Let $p$ be odd.
\begin{enumerate}
\item[(a)] If $u\in\Hur^\times$ satisfies $u\equiv n\pmod{p\Hur}$ with $n\in\mathbb{Z}$, then $u=\pm1$. Hence the classes $\bar\imath,\bar\jmath,\bar k$ of $\varphi(i),\varphi(j),\varphi(k)$ in $\mathrm{PGL}_2(\F_p)$ are three distinct commuting involutions with $\bar\imath\,\bar\jmath=\bar k$, and the class $\bar\rho$ of $\varphi(\rho)$, $\rho=\tfrac12(1+i+j+k)$, conjugates $\bar\imath\to\bar\jmath\to\bar k\to\bar\imath$.
\item[(b)] The centralizer in $\mathrm{PGL}_2(\F_p)$ of the image of $\Hur^\times$ is trivial.
\end{enumerate}
\end{lemma}

\begin{proof}
(a) Every nonzero coordinate of an element of $p\Hur$ has absolute value at least $p/2\geq3/2$, while the imaginary coordinates of a Hurwitz unit have absolute value at most $1$. So $u-n\in p\Hur$ forces the imaginary part of $u$ to vanish, and the only real units are $\pm1$. Applied to $ij^{-1}=-k$ and its companions, this shows $\bar\imath,\bar\jmath,\bar k$ are pairwise distinct and nontrivial; they are involutions since $i^2=j^2=k^2=-1$ is scalar, they commute since $ij=-ji$, and $\rho i\rho^{-1}=j$, $\rho j\rho^{-1}=k$ hold in $\Hur$.

(b) Let $V=\{1,\bar\imath,\bar\jmath,\bar k\}\cong(\mathbb{Z}/2)^2$, and let $g$ centralize the image of $\Hur^\times$; then $g$ centralizes $V$. We claim $C_{\mathrm{PGL}_2(\overline{\F}_p)}(V)=V$, which suffices since $\mathrm{PGL}_2(\F_p)\subseteq\mathrm{PGL}_2(\overline{\F}_p)$. A noncentral involution of $\mathrm{PGL}_2(\overline{\F}_p)$ ($p$ odd) lifts to $s\in\mathrm{GL}_2(\overline{\F}_p)$ with $s^2$ scalar; scaling by a square root makes $s^2=1$ with $s\neq\pm1$, so after conjugation $\bar\imath=\sigma:=[\operatorname{diag}(1,-1)]$. The centralizer of $\sigma$ consists of the diagonal and antidiagonal classes. A diagonal class $[\operatorname{diag}(a,b)]$ is an involution only if $(a/b)^2=1$, i.e.\ it is $1$ or $\sigma$; so $\bar\jmath$ is antidiagonal, and conjugating by a suitable diagonal class (which fixes $\sigma$) makes $\bar\jmath=\tau:=\left[\begin{smallmatrix}0&1\\1&0\end{smallmatrix}\right]$. Now: a diagonal class commutes with $\tau$ iff $[\operatorname{diag}(b,a)]=[\operatorname{diag}(a,b)]$ iff $(a/b)^2=1$, giving $\{1,\sigma\}$; an antidiagonal class $\left[\begin{smallmatrix}0&a\\b&0\end{smallmatrix}\right]$ commutes with $\tau$ iff $(a/b)^2=1$, giving $\{\tau,\sigma\tau\}$. Hence $C(\{\sigma,\tau\})=\{1,\sigma,\tau,\sigma\tau\}=V$, so $g\in V$. But conjugation by $\bar\rho$ permutes $\bar\imath,\bar\jmath,\bar k$ cyclically, so no element of $V\smallsetminus\{1\}$ centralizes $\bar\rho$. Hence $g=1$.
\end{proof}

\begin{proof}[Proof of Theorem~\ref{thm:A}]
Conjugating $PQ=Q'P'$ by a unit $u$ gives
\[
(uPu^{-1})(uQu^{-1})=(uQ'u^{-1})(uP'u^{-1})
\]
with norms preserved, whence
\begin{equation}\label{eq:equiv}
\pi_{uQu^{-1}} \;=\; \rho_u\,\pi_Q\,\rho_u^{-1}\qquad(u\in\Hur^\times).
\end{equation}
Fix a frame. By Lemma~\ref{lem:model}, \eqref{eq:equiv} and faithfulness, condition (i) says $[\varphi(Q)]$ commutes in $\mathrm{PGL}_2(\F_p)$ with $[\varphi(u)]$ for every unit $u$; by Lemma~\ref{lem:centralizer} this holds iff $[\varphi(Q)]=1$, i.e.\ iff $\varphi(Q)$ is scalar, which --- since $\varphi$ is an $\F_p$-algebra isomorphism and the scalars of $\Hur/p\Hur$ are the classes of rational integers --- is condition (iii). And (ii) says $[\varphi(Q)]$ acts trivially on $\mathbb{P}^1(\F_p)$, which by faithfulness is again $[\varphi(Q)]=1$. For the final sentence: on such a set, $I$ is constant on the orbit of a $Q$ with $\pi_Q\neq\mathrm{id}$, while by (i)$\Rightarrow$(ii) the values $\pi_{uQu^{-1}}$ are not all equal.
\end{proof}

\begin{remark}\label{rem:residue}
Identity \eqref{eq:equiv} also delimits the residue: the unit-conjugacy class of $Q$ determines $\pi_Q$ exactly up to conjugation by the group $\{\rho_u:u\in\Hur^\times\}\subseteq\operatorname{Sym}(\mathcal{C}_p)$ --- relabelling by the algebra's own motions --- and Theorem~\ref{thm:A} says this bound is attained whenever it is not vacuous.
\end{remark}

\section{The minimal witness}\label{sec:witness}

Let $p=3$. The four classes of norm-$3$ primes are represented by
\[
P_0=\tfrac12(3+i+j+k),\quad P_1=\tfrac12(3+i+j-k),\quad
P_2=\tfrac12(3+i-j+k),\quad P_3=\tfrac12(3+i-j-k),
\]
which are pairwise non-left-associate (machine check, Remark~\ref{rem:verify}), hence exhaust $\mathcal{C}_3$. Take $q=5$ and
\[
Q\;=\;2+i,\qquad 2+j,\qquad 2+k,\qquad 2-i.
\]
With $\rho=\tfrac12(1+i+j+k)$, using $\rho i\rho^{-1}=j$, $\rho j\rho^{-1}=k$, and $jij^{-1}=-i$:
\[
\rho\,(2+i)\,\rho^{-1}=2+j,\qquad
\rho\,(2+j)\,\rho^{-1}=2+k,\qquad
j\,(2+i)\,j^{-1}=2-i,
\]
so the four primes form one $\Hur^\times$-conjugacy orbit and agree under \emph{every} conjugation-invariant function --- not only norm ($5$), trace ($4$), and $\tr^2-4q=-4$, but all of them at once.

Direct Euclidean refactorization (all sixteen products in Appendix~\ref{app:table}; one specimen below) gives:
\[
\begin{array}{c|c|c}
Q & \pi_Q \text{ in cycle notation} & \big(\pi_Q(P_0),\,\pi_Q(P_1),\,\pi_Q(P_2),\,\pi_Q(P_3)\big)\\\hline
2+i & (P_0\,P_2\,P_3\,P_1) & (P_2,\ P_0,\ P_3,\ P_1)\\
2+j & (P_0\,P_1\,P_2\,P_3) & (P_1,\ P_2,\ P_3,\ P_0)\\
2+k & (P_0\,P_3\,P_1\,P_2) & (P_3,\ P_2,\ P_0,\ P_1)\\
2-i & (P_0\,P_1\,P_3\,P_2) & (P_1,\ P_3,\ P_0,\ P_2)
\end{array}
\]
Reading down any column of values, the four entries are never all equal: the conjugates disagree at $P_0$ (values $P_2,P_1,P_3,P_1$), at $P_1$ ($P_0,P_2,P_2,P_3$), at $P_2$ ($P_3,P_3,P_0,P_0$), and at $P_3$ ($P_1,P_0,P_1,P_2$). The specimen line:
\[
P_0(2+j)=\tfrac12(5+i+5j+3k)
=\underbrace{\tfrac12(3+i+j+3k)}_{\text{norm }5}
\underbrace{\tfrac12(3+i+j-k)}_{=\,P_1,\ \text{norm }3}
\Rightarrow \pi_{2+j}(P_0)=P_1,
\]
verified by expanding the two right-hand factors: their product is $10+2i+10j+6k$.

\begin{proof}[Proof of Theorem~\ref{thm:B}]
Immediate from the orbit computation above, the machine-verified table, and the observation on its columns.
\end{proof}

\begin{example}[equivariance in action]\label{ex:equiv}
The table is not four unrelated accidents; identity \eqref{eq:equiv} generates it. Conjugation by $\rho$ fixes $P_0$ and $3$-cycles the other representatives: $\rho_\rho=(P_1\,P_3\,P_2)$ on classes (machine check). Relabelling the first cycle by $\rho_\rho$,
\[
\rho_\rho\,(P_0\,P_2\,P_3\,P_1)\,\rho_\rho^{-1}
=(P_0\ \rho_\rho P_2\ \rho_\rho P_3\ \rho_\rho P_1)
=(P_0\,P_1\,P_2\,P_3)=\pi_{2+j},
\]
exactly as $\rho(2+i)\rho^{-1}=2+j$ demands; likewise $\rho_j=(P_0\,P_2)(P_1\,P_3)$ conjugates $\pi_{2+i}$ into $\pi_{2-i}$. The invariant layer cannot see these relabellings, and the labelled permutation, beyond its cycle type, consists of nothing else.
\end{example}

\begin{remark}\label{rem:inverse}
$\pi_{2-i}=\pi_{2+i}^{-1}$ in the table, as Lemma~\ref{lem:model} forces: $2-i=\overline{2+i}$.
\end{remark}

\begin{remark}[verification]\label{rem:verify}
The ancillary file \texttt{anc/verify\_witness.py} recomputes, in exact arithmetic: the sixteen refactorizations by direct unit search over the $24$ Hurwitz units (Route 1); the same four permutations independently via the conjugation action on trace-zero lines in $\Hur/3\Hur$, in the model of \cite{FGS} (Route 2), sharing no code path with Route 1; the pairwise non-associativity of $P_0,\dots,P_3$; the conjugacy-orbit identities; the columnwise disagreements; and both instances of \eqref{eq:equiv} in Example~\ref{ex:equiv}. The two routes agree on all four permutations. The computation also pins a convention: in our normalization the line action of Route 2 reads $t\mapsto\bar Q\,t\,\bar Q^{-1}$; the opposite direction returns exactly $\pi_Q^{-1}$.
\end{remark}

\section{Where the remaining information lives}\label{sec:cost}

Theorem~\ref{thm:A} says conjugation invariants stop at the cycle type. Two standard observations price what the rest of $\pi_Q$ costs.

\begin{proposition}[frames form a torsor]\label{prop:torsor}
Any two frames $\varphi,\varphi':\Hur/p\Hur\to M_2(\F_p)$ differ by an inner automorphism of $M_2(\F_p)$, and every class in $\mathrm{PGL}_2(\F_p)$ so occurs: the set of frames is a $\mathrm{PGL}_2(\F_p)$-torsor. Since $\mathrm{PGL}_2(\F_p)$ is sharply $3$-transitive on $\mathbb{P}^1(\F_p)$, the identification $\mathcal{C}_p\cong\mathbb{P}^1(\F_p)$ is canonical only up to M\"obius relabelling: no assignment of coordinates to the classes, not even to three of them, is preferred by the algebra.
\end{proposition}

\begin{proof}
Skolem--Noether applied to $\varphi'\circ\varphi^{-1}$; conversely post-composing a frame with conjugation by $g\in\mathrm{GL}_2(\F_p)$ is again a frame, and trivially so exactly for scalar $g$. Sharp $3$-transitivity is classical.
\end{proof}

\begin{proposition}[one value costs one coset]\label{prop:cost}
Fix a frame and let $g_Q:=[\varphi(Q)]$, so $\pi_Q=R_{g_Q}$ \textup{(}Lemma~\ref{lem:model}\textup{)}. Fix a class $C$ with stabilizer $G_C$ under the right action. Then for $g_1,g_2\in\mathrm{PGL}_2(\F_p)$,
\[
C\,g_1=C\,g_2 \iff G_C\,g_1=G_C\,g_2 ,
\]
so the single value $\pi_Q(C)$ carries exactly the coset $G_C\,g_Q$, no more and no less; and the full labelled $\pi_Q$ carries $g_Q$ itself.
\end{proposition}

\begin{proof}
The right action is transitive, so $g\mapsto Cg$ induces a bijection $G_C\backslash\mathrm{PGL}_2(\F_p)\to\mathbb{P}^1(\F_p)$: orbit--stabilizer. The last clause is faithfulness.
\end{proof}

\begin{example}
For $p=3$: $|\mathrm{PGL}_2(\F_3)|=24$, $|G_C|=6$, so one value selects one of $24/6=4$ cosets --- one of four destinations --- while the full casting selects one of $24$ group elements. Both strictly exceed what invariants supply, which by Theorem~\ref{thm:A} is the conjugacy class of $g_Q$ alone.
\end{example}

The ladder is strict at every rung: invariants determine the cycle type completely \cite{CK,LM} and nothing more (Theorem~\ref{thm:A}); one value is one stabilizer coset (Proposition~\ref{prop:cost}); the labelled permutation is the group element; and there is no distinguished frame in which to read any of it (Proposition~\ref{prop:torsor}).

\section{A motivating picture}\label{sec:pictures}

We record, as motivation and not proof, the drawn presentation from which the witness was extracted; the full development is \cite{Fried}. The multiplications of $\mathbb{Z}[i]$ and $\mathbb{Z}[\omega]$ can be drawn as configurations of signed chords on the coefficient array $\left(\begin{smallmatrix}a&b\\ c&d\end{smallmatrix}\right)$: each chord is a product of two coefficients, signed by how the term enters the expansion (Figure~\ref{fig:legend}); the configurations are the structure constants made visible. The drawn assembly of the quaternion product (Figure~\ref{fig:box}) splits each component of
\begin{align*}
pq &= (a_1a_2-b_1b_2-c_1c_2-d_1d_2)
 + (a_1b_2+b_1a_2+c_1d_2-d_1c_2)\,i \\
 &\quad + (a_1c_2-b_1d_2+c_1a_2+d_1b_2)\,j
 + (a_1d_2+b_1c_2-c_1b_2+d_1a_2)\,k
\end{align*}
into plain and conjugated Gaussian pieces: writing $q=z+wj$ with
$z,w\in\mathbb{Z}[i]$ and using $jz=\bar z\,j$ gives
$q_1q_2=(z_1z_2-w_1\bar w_2)+(z_1w_2+w_1\bar z_2)\,j$, so every chord is a
term of a product in $\mathbb{Z}[i]$, half of them with one factor
conjugated. The splitting elects one imaginary unit to span the complex
plane, and the unit rotation $i\to j\to k$ --- conjugation by
$\rho=\tfrac12(1+i+j+k)$, which satisfies $\rho^2=\rho-1$ and so generates
a copy of $\mathbb{Z}[\omega]$ inside $\Hur$ --- does not preserve that
election. The drawing is a frame; the algebra permutes its frames (Proposition~\ref{prop:torsor}). Theorem~\ref{thm:B} is that instability made arithmetic: it carries $2+i$ through the very motion the drawn assembly refuses to close under, and the labelled permutation moves while every invariant stands still.

\begin{figure}[htbp]
\centering
\includegraphics[width=0.28\textwidth]{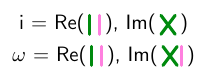}
\caption{The two planar multiplications, drawn. For each of $i$ and $\omega$, the configurations in which the real and imaginary parts of a product assemble: chords are coefficient products, signed by how the term enters the expansion.}
\label{fig:legend}
\end{figure}

\begin{figure}[htbp]
\centering
\includegraphics[width=0.78\textwidth]{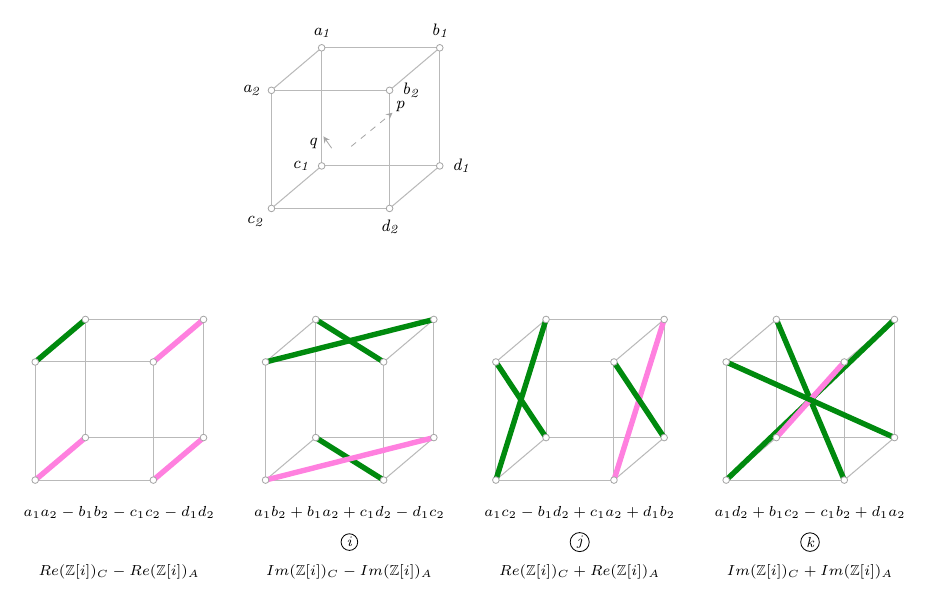}
\caption{The quaternion product split across the cube: each panel is one output component of $pq$; chords are the signed terms of the displayed expansion. Each imaginary component carries exactly one minus, in a different place each time; the assembly is built from plain and conjugated Gaussian multiplications
(legend in Figure~\ref{fig:legend}), and its splitting does not close under
the unit rotation $i\to j\to k$.}
\label{fig:box}
\end{figure}

\appendix

\section{The sixteen refactorizations}\label{app:table}

Each line is an identity in $\Hur$, verifiable by expanding the single quaternion product on the right; the first factor has norm $5$, the second is the listed representative of norm $3$, and the class map is read off the second factor.

\begin{alignat*}{2}
P_0(2+i) = \tfrac12(5+5i+3j+k) &= \tfrac12(3+i+3j+k)\cdot \tfrac12(3+i-j+k) &&\ \Rightarrow\ P_0\mapsto P_2,\\
P_1(2+i) = \tfrac12(5+5i+j-3k) &= \tfrac12(3+i+j-3k)\cdot \tfrac12(3+i+j+k) &&\ \Rightarrow\ P_1\mapsto P_0,\\
P_2(2+i) = \tfrac12(5+5i-j+3k) &= \tfrac12(3+i-j+3k)\cdot \tfrac12(3+i-j-k) &&\ \Rightarrow\ P_2\mapsto P_3,\\
P_3(2+i) = \tfrac12(5+5i-3j-k) &= \tfrac12(3+i-3j-k)\cdot \tfrac12(3+i+j-k) &&\ \Rightarrow\ P_3\mapsto P_1,\\[4pt]
P_0(2+j) = \tfrac12(5+i+5j+3k) &= \tfrac12(3+i+j+3k)\cdot \tfrac12(3+i+j-k) &&\ \Rightarrow\ P_0\mapsto P_1,\\
P_1(2+j) = \tfrac12(5+3i+5j-k) &= (1+2j)\cdot \tfrac12(3+i-j+k) &&\ \Rightarrow\ P_1\mapsto P_2,\\
P_2(2+j) = \tfrac12(7+i+j+3k) &= \tfrac12(3-i+j+3k)\cdot \tfrac12(3+i-j-k) &&\ \Rightarrow\ P_2\mapsto P_3,\\
P_3(2+j) = \tfrac12(7+3i+j-k) &= (2-k)\cdot \tfrac12(3+i+j+k) &&\ \Rightarrow\ P_3\mapsto P_0,\\[4pt]
P_0(2+k) = \tfrac12(5+3i+j+5k) &= (1+2k)\cdot \tfrac12(3+i-j-k) &&\ \Rightarrow\ P_0\mapsto P_3,\\
P_1(2+k) = \tfrac12(7+3i+j+k) &= (2+j)\cdot \tfrac12(3+i-j+k) &&\ \Rightarrow\ P_1\mapsto P_2,\\
P_2(2+k) = \tfrac12(5+i-3j+5k) &= \tfrac12(3+i-3j+k)\cdot \tfrac12(3+i+j+k) &&\ \Rightarrow\ P_2\mapsto P_0,\\
P_3(2+k) = \tfrac12(7+i-3j+k) &= \tfrac12(3-i-3j+k)\cdot \tfrac12(3+i+j-k) &&\ \Rightarrow\ P_3\mapsto P_1,\\[4pt]
P_0(2-i) = \tfrac12(7-i+j+3k) &= \tfrac12(3-i-j+3k)\cdot \tfrac12(3+i+j-k) &&\ \Rightarrow\ P_0\mapsto P_1,\\
P_1(2-i) = \tfrac12(7-i+3j-k) &= \tfrac12(3-i+3j+k)\cdot \tfrac12(3+i-j-k) &&\ \Rightarrow\ P_1\mapsto P_3,\\
P_2(2-i) = \tfrac12(7-i-3j+k) &= \tfrac12(3-i-3j-k)\cdot \tfrac12(3+i+j+k) &&\ \Rightarrow\ P_2\mapsto P_0,\\
P_3(2-i) = \tfrac12(7-i-j-3k) &= \tfrac12(3-i+j-3k)\cdot \tfrac12(3+i-j+k) &&\ \Rightarrow\ P_3\mapsto P_2.
\end{alignat*}

The four permutations, on $(P_0,P_1,P_2,P_3)$:
\[
(P_0\,P_2\,P_3\,P_1),\quad (P_0\,P_1\,P_2\,P_3),\quad (P_0\,P_3\,P_1\,P_2),\quad (P_0\,P_1\,P_3\,P_2).
\]

\end{document}